\documentclass[conference]{IEEEtran}
\usepackage[utf8]{inputenc}
\usepackage{amsmath}
\usepackage{graphicx}
\usepackage{caption}
\usepackage{epstopdf}
\usepackage{amsmath,amsfonts,amssymb,amsthm,epsfig,epstopdf,url,array}
\usepackage[linesnumbered,boxed,ruled]{algorithm2e}
\usepackage{graphicx}
\usepackage{color}
\usepackage{textcomp}

\usepackage{cite}
\usepackage{amsthm}
\usepackage{amsmath,amsthm}
\usepackage{marvosym}
\usepackage{array}
\def\BibTeX{{\rm B\kern-.05em{\sc i\kern-.025em b}\kern-.08em
    T\kern-.1667em\lower.7ex\hbox{E}\kern-.125emX}}

\newtheorem{myth}{Theorem}
\newtheorem{mydef}[myth]{Definition} 
\newtheorem{myprop}[myth]{\bf Proposition}

\newtheorem{myle}[myth]{\bf Lemma}

\newcommand{\cN}{\mathcal{N}}

\newcommand{\cS}{\mathcal{S}}

\ifx\pdftexversion\undefined
\usepackage{hyperref}
\else
\usepackage[pdftex]{hyperref}
\fi
\hypersetup{
	colorlinks=true,
	linkcolor=blue,
	citecolor=blue,
	urlcolor=blue
}
\newcommand{\enum}[1]{\textit{#1)}}

\begin{document}
\IEEEoverridecommandlockouts
\title{Joint System Latency and Data Freshness Optimization for Cache-enabled Mobile Crowdsensing Networks}
\author{
\IEEEauthorblockN{
Kexin Shi\IEEEauthorrefmark{1}, Yaru Fu\IEEEauthorrefmark{1}, Yongna Guo\IEEEauthorrefmark{2}, Fu Lee Wang\IEEEauthorrefmark{1}, and Yan Zhang\IEEEauthorrefmark{3}}

\IEEEauthorblockA{\IEEEauthorrefmark{1}\text{School of Science and Technology, Hong Kong Metropolitan University, Hong Kong}} 
\IEEEauthorblockA{\IEEEauthorrefmark{2}\text{Department of Electrical Engineering and Computer Science, KTH Royal Institute of Technology, Sweden}}
\IEEEauthorblockA{\IEEEauthorrefmark{3}\text{Department of Informatics, University of Oslo, Norway}}
}
\maketitle
\begin{abstract}
Mobile crowdsensing (MCS) networks enable large-scale data collection by leveraging the ubiquity of mobile devices. However, frequent sensing and data transmission can lead to significant resource consumption. To mitigate this issue, edge caching has been proposed as a solution for storing recently collected data. Nonetheless, this approach may compromise data freshness. In this paper, we investigate the trade-off between re-using cached task results and re-sensing tasks in cache-enabled MCS networks, aiming to minimize system latency while maintaining information freshness. To this end, we formulate a weighted delay and age of information (AoI) minimization problem, jointly optimizing sensing decisions, user selection, channel selection, task allocation, and caching strategies. The problem is a mixed-integer non-convex programming problem which is intractable. Therefore, we decompose the long-term problem into sequential one-shot sub-problems and design a framework that optimizes system latency, task sensing decision, and caching strategy subproblems. When one task is re-sensing, the one-shot problem simplifies to the system latency minimization problem, which can be solved optimally. The task sensing decision is then made by comparing the system latency and AoI. Additionally, a Bayesian update strategy is developed to manage the cached task results. Building upon this framework, we propose a lightweight and time-efficient algorithm that makes real-time decisions for the long-term optimization problem. Extensive simulation results validate the effectiveness of our approach.
\end{abstract}

\section{Introduction}
The rapid advancement of the Internet of Things (IoT) and artificial intelligence (AI) has fueled an unprecedented demand for real-time, large-scale data sensing and analysis \cite{wenshuai}. Mobile crowdsensing (MCS) has emerged as a transformative paradigm that leverages the ubiquity of mobile devices to collect and utilize data on a massive scale \cite{mcs1}.  The efficacy of MCS systems is intrinsically tied to the active engagement of their user base. Various incentive mechanisms have been developed to encourage and maintain user engagement \cite{FL+MCS,incentive_V}. Moreover, to effectively utilize users' limited resources, various task allocation strategies have been proposed, aiming to minimize system costs while ensuring satisfactory task coverage and quality \cite{us_4,us_5}. However, these studies often overlook the practical constraints of wireless networks and their impact on MCS performance.

Considering the sensing and communication processes in MCS systems, several studies have examined the impact of limited wireless network resources on MCS performance. For instance, the study in \cite{Li} explored the optimization of sensing rewards by simultaneously considering task allocation, user selection, and constrained energy and transmission power allocation. Authors in \cite{WCNC2024} addressed the challenge of minimizing latency, considering practical constraints such as task size requirements and subband allocation limitations. To further optimize resource utilization, some studies have integrated edge caching into MCS networks, enabling local storage of frequently accessed data to reduce sensing and transmission demands. For example, the work in \cite{mcs+cache2} examined the participant selection process in edge-cached systems, proposing a caching strategy based on data quality and developing incentives to promote high-quality data contributions. Researchers in \cite{mcs+cache3} introduced a priority-based cache management technique to enhance network communication efficiency in cache-enabled MCS environments. While edge caching can effectively manage network resources and reduce energy consumption, it may introduce potential data freshness issues in MCS systems. The timeliness of cached sensing tasks is crucial in many MCS applications, as the sensing data is highly time-sensitive. To quantify and address this data freshness challenge, Age of Information (AoI) is a widely adopted metric for quantifying information freshness, prompting recent works to consider the freshness of sensing task data in MCS systems. For instance, studies in \cite{paper1,paper8} explored data freshness in MCS through the lens of incentives and pricing strategies. Besides, other research \cite{paper3,paper6} examined information freshness thresholds in the context of UAV-supported MCS networks. 

Maintaining information freshness often requires frequent sensing and transmission, which can strain network resources and limit the scalability of MCS deployments. However, existing studies \cite{paper1,paper8,paper3,paper6} on information freshness in MCS systems have not adequately addressed the impact of wireless network resource constraints on the practicality of their solutions. To bridge this gap, this paper considers the wireless communication constraints and investigates the trade-off between re-using cached task results and re-sensing tasks in cache-enabled MCS networks.  We propose a time-efficient framework that intelligently balances this trade-off to minimize system latency while ensuring the freshness of cached information. Our main contributions are summarized as follows:
\begin{itemize}
\item We formulate a weighted system latency and AoI minimization problem in cache-enabled MCS networks that jointly optimizes the sensing decision, user selection, subchannel allocation, task allocation, and caching strategy. The formulated problem is a mixed-integer non-convex programming problem, which is challenging to solve. To facilitate the analysis, we decompose the original problem into multiple one-shot sub-problems.
\item Each one-shot sub-problem remains a mixed-integer non-convex problem and is thus intractable. We further decompose it into system latency minimization, task sensing, and caching decision subproblems. By assuming task re-sensing, the system latency minimization sub-problem can be solved optimally using the Hungarian algorithm. Subsequently, task sensing decisions are made by balancing system latency against AoI. Finally, we propose a Bayesian update strategy to manage the cached task results. 
\item Building upon the one-shot optimization algorithm, we develop an overall time-efficient framework to address the original long-term optimization problem. Extensive numerical results are presented to validate the superiority of our developed scheme compared to various benchmark strategies. 
\end{itemize}
The subsequent sections of this paper are structured as follows: Section II introduces the network description and system model. Section III gives the joint optimization problem. Section IV presents the proposed time-efficient framework in detail. Section V demonstrates the superiority of our developed algorithm through numerical simulations. Finally, Section VI summarizes our work.
\section{System Model}
\begin{figure}[t]
\centering
\includegraphics[width=8cm]{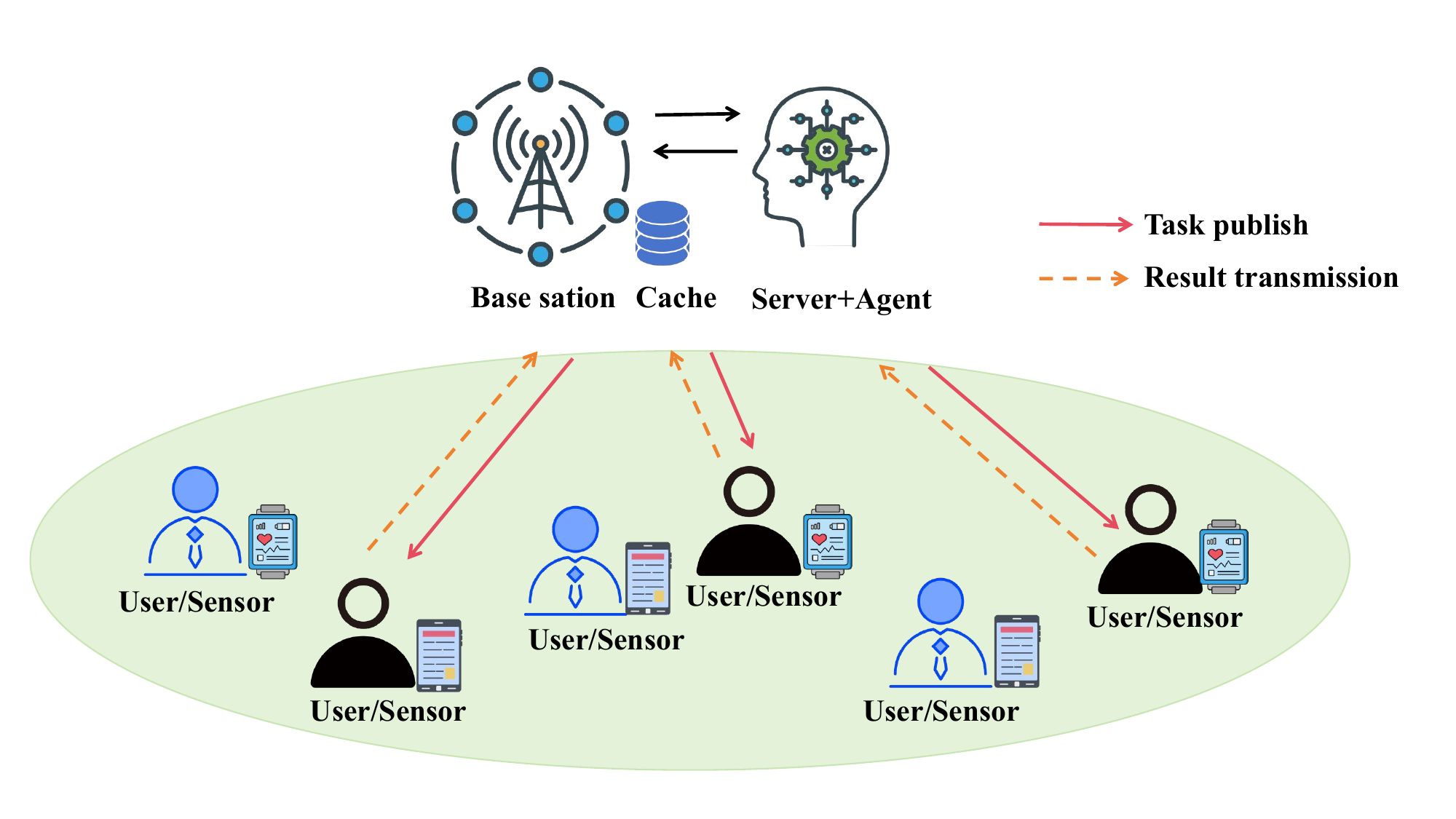}
\captionsetup{font=footnotesize}
\caption{An example of the cache-enabled MCS systems.}\label{model}
\end{figure}
\subsection{Network Description and AoI Mechanism}
As illustrated in Figure \ref{model},  we consider a cache-enabled MCS network where a cellular base station (BS) integrates an edge server and a cache device with capacity $Q$. The network comprises $K$ sensor-equipped users, denoted by $\mathcal{K} = \{1,2, \ldots, K\}$. The system operates in discrete time slots $\mathcal{T} = \{1,2, \ldots, T\}$, with interval $\delta_t$ for time slot $t$. In each time slot, the MCS agent publishes a sensing task $i_t$ from the task type $\mathcal{M} = \{1,2, \ldots, M\}$. The task sensing decision at time slot $t$ is represented by $l_{i_t}^t \in \{0, 1\}$, where $l_{i_t}^t = 1$ indicates re-sensing and $l_{i_t}^t = 0$ means re-using cached data. The MCS agent decides whether to request users to re-sense task $i_t$ or re-use cached results from the BS, based on the system latency and the freshness of cached results. For re-sensing scenarios, the agent selects users through binary variable $s_k^t$ for user $k$ at time slot $t$, considering spectrum resource constraints. At the end of each time slot, the MCS agent decides whether to cache the sensing results at the BS.

Denote the size of task type $i$ as $V_i$. Let $c_{i}^t \in \{0,1\}$ represents the caching state of task $i$ at time slot $t$, following the capacity constraint $\sum_{i \in \mathcal M }c_i^{{t}} V_{i} \le Q$. To ensure data freshness, we evaluate the cached task results using AoI, denoted as $\Delta_{i}^t$ for task $i$ at the beginning of time slot $t$. A maximum AoI threshold $\Delta_{\rm max}$ is set for the BS cache. If $\Delta_{i}^t \ge \Delta_{\rm max}$, the cached task $i$ is discarded, i.e., $c_{i}^{t}=0$. In time slot $t$, if task $i$ was already cached in the previous time slot $t-1$, its AoI is increased by the time duration $\delta_{t - 1}$. Conversely, if task $i$ was not cached in time slot $t-1$, its AoI is reset to 0, which can be expressed as:
\begin{equation}
\Delta_{i}^t = (\Delta_{i}^{t-1} + \delta_{t-1}) \cdot c_{i}^{t-1}, ~i \in \mathcal M.
\label{eq:aoi}
\end{equation}
\subsection{Task Sensing and Transmission Model}
When user $k$ is selected to sense task $i_t$ at time slot $t$, this user will take on a certain amount of sensing data size $z_k^t$. To ensure task completion at time slot $t$, we have:
\begin{equation}
\sum_{k\in \mathcal{K}} z_k^t \ge l_{i_t}^t V_{i_t}.  
\label{eq:task}
\end{equation}
 Denote $o_k^t$ as the sensing data rate for user $k$ at time slot $t$, the sensing latency can be calculated by:
\begin{equation}
T_{k,{\text{sen}}}^t = \frac{z_k^t}{o_k^t}.
\end{equation}
 Let $e_k^t$ be the sensing energy consumption per bit for user $k$ at time slot $t$, the sensing energy consumption is given by:
 \begin{equation}
 E_{k,{\text{sen}}}^t = e_k^t \cdot z_k^t.
 \end{equation}
For data transmission, the system employs orthogonal frequency division multiple access (OFDMA) with $N$ orthogonal subchannels of bandwidth $W$. The subchannel allocation variable $b_{k,n}^t \in \{0, 1\}$ indicates whether subchannel $n$ is assigned to user $k$ at time slot $t$, where $n \in \{0,1,\cdots,N\}.$ The relationship between $b_{k,n}^t$ and the user selection variable $s_k^t$ is given by $\sum_{n\in \mathcal{N}} b_{k,n}^t= s_k^t,~k\in{\mathcal{K}}$. Furthermore, each subchannel can be allocated to at most one user, i.e., $\sum_{k\in \mathcal{K}} b_{k,n}^t \le 1, ~n\in \mathcal{N}$. Let $P_k^t$ be the transmission power of user $k$ at time slot $t$. Denote the channel gain of user $k$ on subchannel $n$ as $g_{k,n}^t$, and $N_0$ represents the noise power density. The achievable transmission rate for user $k$ is given by:
\begin{equation}
r_k^t = \sum_{n\in \mathcal{N}} b_{k,n}^t W \log_{2}\left(1+\frac{P_k^t g_{k,n}^t}{N_0 W}\right).
\label{eq:rkt}
\end{equation}
 Similarly, the transmission latency  for user $k$ at time slot $t$ can be calculated by:
\begin{equation}
T_{k,{\text {tran}}}^t = \frac{z_k^t}{r_k^t}.
\end{equation}
The energy consumption for data transmission can be given by:
\begin{equation}
E_{k,{\text {tran}}}^t = P_k^t \cdot T_{k,{\text {tran}}}^t.
\end{equation}
Furthermore, the total energy consumption for each user is bounded by:
\begin{equation}
E_{k,{\text{sen}}}^t +E_{k,{\text{tran}}}^t \le E_k, ~k\in \mathcal{K}.
\label{eq:Elimit}
\end{equation}
Overall, we denote the maximum system latency for each time slot $t$ as ${\bar D}^t$, which is determined by the maximum value of the completion of data sensing and aggregation of all selected users, can be expressed as follows:
        \begin{equation}
		{\bar D}^t =  \max\limits_{ k \in \mathcal{K}} \left\{  T_{k,{\text{sen}}}^t + T_{k,{\text {tran}}}^t\right\}.
        \label{eq:D_bar}
	\end{equation}

\section{Problem Formulation}

In this work, we aim to minimize the weighted AoI and system latency optimization problem under users' energy consumption and cache capacity constraints. We define task sensing decision, bandwidth allocation, user selection decision, sensing data size allocation, and caching strategy for time slot $t$ as $\boldsymbol l^t = (l_{i_t}^t)_{i_t \in \mathcal M}$, $\boldsymbol b^t= (b_{k,n}^t)_{k\in \mathcal{K}, n\in \mathcal{N}}$, $\boldsymbol s^t = (s_{k}^t)_{k\in \mathcal{K}}$, $\boldsymbol z^t = (z_k^t)_{k\in \mathcal{K}}$, $\boldsymbol c^t = (c_i^t)_{i\in \mathcal{M}}$, respectively. Furthermore, we give the task sensing vector, bandwidth allocation vector, user selection vector, sensing data size allocation, and caching strategy over all time slots as $\boldsymbol l = (\boldsymbol l^t)_{t\in \mathcal T}$, $\boldsymbol b = (\boldsymbol b^t)_{t\in \mathcal T}$, $\boldsymbol s = (\boldsymbol s^t)_{t\in \mathcal T}$, $\boldsymbol z = (\boldsymbol z^t)_{t\in \mathcal T}$, $\boldsymbol c = (\boldsymbol c^t)_{t\in \mathcal T}$, respectively. Moreover, we set $\beta_{1}$ and $\beta_{2}$ as the weights for system latency and AoI, respectively. The optimization problem can be mathematically formulated as follows:
	\begin{equation}
		\min\limits_{\boldsymbol {l, b, s, z, c}}~  \frac{1}{T}  \sum_{t \in \mathcal{T}} \left(\beta_1 l_{i_t}^t {\bar D}^t + \beta_2 (1 - l_{i_t}^t) \Delta_{i_t}^t\right)
        \label{eq:OF}
	\end{equation} 
\begin{align*}
  \text{ s.t. }  &C_1: \sum_{k\in \mathcal{K}} z_k^t \ge l_{i_t}^t V_{i_t}, ~t\in \mathcal{T},\\
    &C_2: \sum_{i \in \mathcal M }c_i^{{t}} V_{i} \le Q, ~t \in \mathcal{T},\\
    &C_3: \sum_{k\in \mathcal{K}} b_{k,n}^t \le 1, ~n\in \mathcal{N}, ~t\in \mathcal{T},\\
    &C_4: \sum_{n\in \mathcal{N}} b_{k,n}^t= s_k^t,~k\in{\mathcal{K}}, ~t\in \mathcal{T},\\
    &C_5: 1- l_{i_t}^t \le c_{i_t}^{t-1}, ~t\in \mathcal{T}, \\
    &C_6: E_{k,\text{sen}}^t + E_{k,\text{tran}}^t \le E_k, ~k\in \mathcal{K}, ~t\in \mathcal{T},\\
    &C_7: 0 \le \Delta_{i}^t \le \Delta_{\text{max}}, ~i\in \mathcal M, ~t\in \mathcal{T},\\
    &C_8: l_{i_t}^t \in \{0,1\}, ~t\in \mathcal{T},\\
    &C_{9}: c_i^t \in \{0,1\}, ~i \in \mathcal{M},  ~t\in \mathcal{T},\\
    &C_{10}: b_{k,n}^t \in \{0,1\}, ~k\in \mathcal{K}, ~n\in \mathcal{N}, ~t\in \mathcal{T},\\
    &C_{11}: s_{k}^t \in \{0,1\}, ~k\in \mathcal{K}, ~t\in \mathcal{T},
\end{align*} 
where $C_1$ specifies the sensing bit size requirement, and $C_2$ limits BS cache capacity. $C_3$ and $C_4$ constrain subchannel allocation. $C_5$ enforces re-sensing for tasks not cached in slot $t-1$. $C_6$ restricts user energy consumption. $C_7$ triggers cache removal when task AoI exceeds the threshold. $C_8-C_{11}$ define binary variables for task sensing, caching, subchannel allocation, and user selection decisions. 

\section{Algorithm Design}
\subsection{Problem Transformation}
Problem \eqref{eq:OF} is a non-convex mixed-integer problem that requires complete system information for optimization, which is often intractable. Therefore, we can decompose it into one-shot sub-problems, which can be expressed as follows:
        \begin{equation}
	\min\limits_{\boldsymbol{l}^t, \boldsymbol{b}^t, \boldsymbol{s}^t, \boldsymbol{z}^t, \boldsymbol{c}^t}~  \beta_1 l_{i_t}^t {\bar D}^t + \beta_2 (1 - l_{i_t}^t) \Delta_{i_t}^t
        \label{eq:one-shot-OF}
        \end{equation}
     \begin{align*}
   \text{ s.t. } &C^*_1: \sum_{k\in \mathcal{K}} z_k^t \ge l_{i_t}^t  V_{i_t},\\
    &C^*_2: \sum_{i \in \mathcal M }c_i^{{t}} V_{i} \le Q,\\
    &C^*_3: \sum_{k\in \mathcal{K}} b_{k,n}^t \le 1, ~n\in \mathcal{N}, \\
    &C_4^*: \sum_{n\in \mathcal{N}} b_{k,n}^t= s_k^t,~k\in{\mathcal{K}},\\
    &C^*_5: 1- l_{i_t}^t \le c_{i_t}^{t-1},\\
    &C^*_6: E_{k,\text{sen}}^t + E_{k,\text{tran}}^t \le E_k, ~k\in \mathcal{K},\\
    &C^*_7: 0 \le \Delta_{i}^t \le \Delta_{\text{max}}, ~i \in \mathcal M,\\
    &C^*_8: l_{i_t}^t \in \{0,1\},\\
    &C^*_{9}: c_i^t \in \{0,1\}, ~i \in \mathcal{M},\\
    &C^*_{10}: b_{k,n}^t \in \{0,1\}, ~k\in \mathcal{K}, ~n\in \mathcal{N},\\
    &C_{11}^*: s_{k}^t \in \{0,1\}, ~k\in \mathcal{K}.
     \end{align*}
Solving problem \eqref{eq:one-shot-OF} is still challenging due to its non-convex, mixed-integer, and nonlinear characteristics. Therefore, we develop a framework to address this one-shot problem efficiently.
\subsection{System Latency Minimization Sub-problem}
Given the MCS agent's decision to re-sense task $i_t$, i.e., $l_{i_t}^t = 1$, the one-shot problem \eqref{eq:one-shot-OF} can be degenerated into the following system latency minimization sub-problem:
        \begin{equation}
	\min_{ \boldsymbol{b}^t, \boldsymbol{s}^t, \boldsymbol{z}^t}~  {\bar D}^t
       \label{eq:minD_bar}
       \end{equation} 
   \begin{align*}
  \text{ s.t. }  &C^*_1: \sum_{k\in \mathcal{K}} z_k^t \ge l_{i_t}^t  V_{i_t},\\
    &C^*_3: \sum_{k\in \mathcal{K}} b_{k,n}^t \le 1, ~n\in \mathcal{N}, \\
    &C_4^*: \sum_{n\in \mathcal{N}} b_{k,n}^t= s_k^t,~k\in{\mathcal{K}},\\
   &C^*_6: E_{k,\text{sen}}^t + E_{k,\text{tran}}^t \le E_k,~k\in \mathcal{K},\\
   &C^*_{10}: b_{k,n}^t \in \{0,1\}, ~k\in \mathcal{K}, ~n\in \mathcal{N},\\   
   &C_{11}^*: s_{k}^t \in \{0,1\}, ~k\in \mathcal{K}.
  \end{align*}
  
 We first assume the user selection variable $\boldsymbol{s}^t$ and subchannel allocation $\boldsymbol b^t$ are fixed. Denote $\mathcal S^t \subset \mathcal K$ as the selected users set at time slot $t$, and let $\alpha_k^t = \frac{1}{o_k^t} + \frac{1}{r_k^t}$ represent the data processing time per bit for user $k \in \cS^t$ at $t$. Then, the sub-problem \eqref{eq:minD_bar} can be reformulated as follows:
  \begin{equation}
  \min\limits_{\boldsymbol z^t} \max\limits_{k \in \mathcal S^t} ~\alpha_{k}^t z_k^t
  \label{eq:minaz}
  \end{equation}
   \begin{align*}
  \text{ s.t. } C_1^*, C_6^*.
    \end{align*}
Let $\hat{z}_k^t$ denote the optimal sensing data size allocation for user $k$ at time slot $t$. The following Lemma characterizes this optimal allocation:
   \begin{myle}
   \label{l1}
    The closed-form optimal solution of minimization problem \eqref{eq:minaz} can be characterized as:
    \begin{equation}
    \hat{z}_k^t = \min\left\{\frac{E_k}{A_{k}^{t}(b_{k,n}^t)}, \frac{V_{i_t}}{\alpha_k^t \left(\sum_{k \in \mathcal {S}^t} \frac{1}{\alpha_k^t}\right)}\right\}, ~ k \in \mathcal{S}^t.
    \label{eq:optzkt}
    \end{equation}
   \end{myle}

   \begin{proof}
   Problem \eqref{eq:minaz} is to minimize the maximum processing time among all selected users. The optimal solution can be obtained by satisfying three essential conditions. First, $C_1^*$ is satisfied with strict equality, i.e., $\sum_{k\in \mathcal{S}^t}z_k^t = l_{i_t}^t V_{i_t}$. Denote the data processing time of user $k$ in time slot $t$ as $\gamma^t_k = \alpha_k^t z_k^t,~ \forall k \in \mathcal{S}^t$. Then, we can get the following equation:
    \begin{equation}
   \sum_{k \in \mathcal S^t}\frac{\gamma^t_k}{\alpha_k^t}= V_{i_t}.
   \label{eq:C1*}
   \end{equation}
   Second, the data processing duration must be uniform across all participating users, i.e., $\gamma^t_k = \alpha_k^t z_k^t = \frac{V_{i_t}}{\sum_{k \in \mathcal S^t} \frac{1}{\alpha_k^t}}$, $\forall k \in \mathcal S^t$, we can get:
\begin{equation}
    {z}_k^t = \frac{V_{i_t}}{\alpha_k^t \left(\sum_{k \in \mathcal {S}^t} \frac{1}{\alpha_k^t}\right)}, ~k \in \mathcal{S}^t.
    \label{eq:optzkt1}
\end{equation}
   Finally, considering the energy constraints $C_6^*$, we obtain that: 
  \begin{equation}
   z_k^t \left( e_k^t + \frac{P_k^t}{\sum_{n\in \mathcal{N}}  b_{k,n}^t W \log_{2}\left(1+\frac{P_k^t g_{k,n}^t}{N_0 W}\right)}\right) \leq E_k, ~k \in \mathcal{S}^t.
   \label{eq:C6*}
   \end{equation}
   Let $A_{k}^{t}(b_{k,n}^t) = e_k^t + \frac{P_k^t}{\sum_{n\in \mathcal{N}}  b_{k,n}^t W \log_{2}\left(1+\frac{P_k^t g_{k,n}^t}{N_0 W}\right)}$, where $k \in \mathcal{S}^t$, we then get:
   \begin{equation}
   {z}_k^t \leq \frac{E_k}{A_{k}^{t}(b_{k,n}^t)}, ~k \in \mathcal{S}^t.
   \label{eq:optzkt2}
   \end{equation}
Combining \eqref{eq:optzkt1} and \eqref{eq:optzkt2}, the Lemma has been proved.
   \end{proof}
   To solve the system latency sub-problem \eqref{eq:minD_bar}, we first address the optimal user-subchannel pairing. Define $k_n \in \mathcal{K}$ as the user assigned to subchannel $n \in \mathcal{N}$. The optimal assignment follows Theorem 2: 
   \begin{myth}
   \label{t2}
    The system latency minimization optimization problem \eqref{eq:minD_bar} is equivalent to:
   \begin{equation}
   \min_{k_1,k_2,...,k_N}  \frac{V_{i_t}}{\sum_{n \in \mathcal N} \frac{1}{\alpha_{k_n}^t}}
   \label{eq:Hungarian}
   \end{equation} 
   \begin{align*}
 \text{\rm s.t. }  C_3^* - C_4^*.
  \end{align*}
\end{myth}

\begin{proof}
Based on Lemma 1, we establish the optimal sensing data size allocation $\hat{z}_k^t$ for each possible user-subchannel assignment in problem \eqref{eq:minaz}. Then, the original problem \eqref{eq:minD_bar} is reduced to finding the optimal user-subchannel assignment $\{k_1, k_2, \ldots, k_N\}$ that minimizes the objective function. This transformation preserves optimality while simplifying the problem structure to a pure assignment problem with constraints $C^*_3$ and $C^*_4$.
\end{proof}
To solve problem \eqref{eq:Hungarian}, we consider a bipartite graph $G = (\mathcal{K} \times \mathcal{N}, E)$, where $\mathcal{K}$ denotes the set of users, $\mathcal{N}$ represents the set of subchannels, $E$ comprises the edges connecting users to subchannels. Each edge $(k, n) \in E$ is weighted by $\frac{1}{\alpha_{k_n}^t}$, corresponding to the terms in the objective function of problem \eqref{eq:Hungarian}. This graph-theoretic formulation leads to Proposition 3:
\begin{myprop}
\label{p3}
The optimal solution to problem \eqref{eq:minD_bar} can be obtained by solving the maximum weighted matching problem in bipartite graph $G$ using the Hungarian algorithm.
\end{myprop}
\begin{proof}
The optimization problem \eqref{eq:Hungarian} is equivalent to finding a matching in a bipartite graph $G$. This formulation precisely corresponds to the maximum weighted matching problem in bipartite graphs, for which the Hungarian algorithm provides an optimal solution. 
\end{proof}
\subsection{Task Sensing and Caching Sub-problem}
\subsubsection{Task sensing decision}
The objective function \eqref{eq:one-shot-OF} comprises two components: system latency $\bar{D}^t$ and AoI value $\Delta_{i_t}^t$, weighted by $\beta_1 l_{i_t}^t$ and $\beta_2 (1 - l_{i_t}^t)$, respectively. Therefore, the weighted minimization problem can be expressed as $\min\{\beta_1 {\bar D}^t , \beta_2 \Delta_{i_t}^t\}$. Let $\beta = \beta_0 \cdot \frac{\beta_1}{\beta_2}$, where $\beta_0 \in [0,1]$ denotes the system re-sensing frequency. Thus, the task sensing decision can be made according to the following Definition:
\begin{mydef}
\label{d4}
Task $i_t$ requires re-sensing if either of the following conditions is satisfied:
\begin{equation}
\left( 0 \le \Delta_{i_t}^t \le \Delta_{\rm{max}} \text{ and } \Delta_{i_t}^t \ge \beta \bar{D}^t \right) \quad \text{or} \quad c_{i_t}^{t-1} = 0.
\end{equation}
\end{mydef}
\subsubsection{Caching policy}
After task sensing completion, the MCS agent determines whether to cache the results, let $\Theta^t$ be the set of cached tasks at time slot $t$. The caching policy follows three scenarios: \enum{1}. First, when the BS cache has not reached its capacity limitation, new task sensing results will directly cached. \enum{2}. Second, for a full cache, incoming updates of existing tasks supersede their previous versions. \enum{3}. Otherwise, we employ Bayesian updating with posterior probability to evaluate the task values. In the cache replacement process, we prioritize removing tasks with lower posterior probabilities. The posterior probability of task $i$ at time slot $t$ is given by:
    \begin{equation}
        X_{{\rm post},i}^t = \frac{X_{{\rm prior},i}^t \cdot L_i^t}{\sum_{i \in \Theta^t} X_{{\rm prior},i}^t \cdot L_i^t}, ~i \in \Theta^t,
        \label{xpost}
    \end{equation}
where $X_{{\rm prior},i}^t = 1/\Delta_i^t$ is the prior probability. The likelihood is calculate by $L^t_i = \ln (1+ \epsilon \cdot F_i^t)$, where $\epsilon$ is a proportional variable set to $1/V_i$, and $F_i^t$ is represents the published frequency of task $i$ in the past. 
\subsection{Time-efficient Joint Optimization Algorithm}
  In this subsection, we present the time-efficient joint optimization algorithm for problem \eqref{eq:OF}. The original problem \eqref{eq:OF} is decomposed into $T$ sequential sub-problems, each solved by the framework established in the previous subsection. Specifically, for the sensing latency minimization sub-problem \eqref{eq:minD_bar}, Lemma \ref{l1} determines the optimal sensing data size for selected users, while Theorem \ref{t2} and Proposition \ref{p3} transform the user-subchannel allocation into a weighted bipartite matching problem solvable via the Hungarian algorithm. Based on Definition \ref{d4}, sensing decisions are made by comparing system latency $\bar{D}^t$ with AoI value $\Delta_{i_t}^t$. Then the caching policy using Bayesian update is given. The complete procedure of our proposed joint optimization algorithm for problem \eqref{eq:OF} is given by Algorithm \ref{a2}. Besides, to evaluate the efficiency of our proposed algorithm, we analyze the computational complexity in the following Theorem:  
 
\begin{algorithm}[t]
    \footnotesize
    \setcounter{AlgoLine}{0}
    \caption{Time-efficient joint optimization algorithm}
    \label{a2}
    \SetKwInOut{Input}{Input}
    \SetKwInOut{Output}{Output}
    \KwIn{Number of total time slots $T$.}
\For{$t \le T$}
{
    Obtain $\boldsymbol{b}^t$, $\boldsymbol{s}^t$, $\boldsymbol{z}^t$ by Proposition \ref{p3}.\\
    Calculate $\Delta_{i_t}^t$ and ${\bar D}^t$ by \eqref{eq:aoi} and \eqref{eq:D_bar}, respectively.\\
    \uIf{$c_{i_t}^{t-1} = 1$ \textbf{and} $\Delta_{i_t}^t \le \beta {\bar D}^t$}{
            Use the old task results in BS cache, i.e., $l_{i_t}^t = 0$.\\
            Reset $\boldsymbol{b}^t =\boldsymbol{0}$, $\boldsymbol{s}^t=\boldsymbol{0}$ and $\boldsymbol{z}^t=\boldsymbol{0}$.
    }
    \Else{
            Re-sensing task $i_t$, i.e., $l_{i_t}^t = 1$.\\
           \uIf{$\boldsymbol{c}^{t-1} V_i + V_{i_t} \le Q$}
    {
     Add the task sensing result $i_t$ to the BS cache, i.e., ${c}_{i_t}^t = 1$.
    }
    \uElseIf{${c}_i^{t-1} V_i = Q$ \textbf{and} ${c}^{t-1}_{i_{t}}=1$}
{
Replace the old version of task $i_t$ with the new sensing task result, update ${c}^t_{i_{t}} = 1$.
}

    \Else{
        Add the task sensing result $i_t$ to the BS cache.
         
        \For{$i \in \Theta^t$}
        {
          Calculate the posterior probability of task $i$ by \eqref{xpost}.
        }
        Sort each task $i \in \Theta^t$ by $X_{{\rm post},i}^t$ in descending order.
        
        Remove the task with the smallest $X_{{\rm post},i}^t$ from the BS cache to satisfy the cache capacity constraint $Q$.

    }
    }
    $t = t + 1$;\\
}
\Output{$\boldsymbol{b}$, $\boldsymbol{s}$, $\boldsymbol{z}$, $\boldsymbol{l}$, $\boldsymbol{c}$.}
\end{algorithm}
\begin{myth}
The computational complexity of Algorithm 1 is $\mathcal{O}(T(N^2 K + M {\log} M))$.
\end{myth}
\begin{proof}
Our proposed algorithm complexity comprises two major components. First, the Hungarian algorithm for user-subchannel matching requires $\mathcal{O}(N^2 K)$ operations \cite{Li}. Second, the caching strategy involves sorting tasks by posterior probabilities with complexity $\mathcal{O}(M {\log} M)$ \cite{OLSA}. Therefore, the total computational complexity of Algorithm 1 is $\mathcal{O}(T(N^2 K + M {\log} M))$.
\end{proof}

\section{Simulation Result}
In this section, we conduct extensive simulations to assess the performance of our proposed time-efficient joint optimization algorithm. The key simulation parameters are presented in Table I. To thoroughly evaluate the performance of our proposed algorithm, we compare it against five baseline schemes:
\begin{itemize}

\item\textbf{Baseline 1}: Assigns each subchannel to the user with the highest channel gain and uniformly allocates tasks among scheduled users, i.e., $z_k^t = \frac{V_{i_t}}{N}$. Employs random task sensing decision and replaces the oldest cached task result when the cache is full.

\item\textbf{Baseline 2}: Randomly assigns subchannels to users and employs a fractional task allocation strategy based on channel gains, where $z_k^t = V_{i_t} \frac{g_{{k_n},n}^t}{\sum_{j \in \cN} g_{{k_j},j}^t}$. Task sensing and caching policies are the same as Baseline 1.

\item\textbf{Baseline 3}: Follows the same subchannel allocation, task sensing, and caching decisions as Baseline 1, with a fractional task allocation approach similar to Baseline 2.

\item\textbf{Baseline 4}: Adopts the same subchannel and task allocation strategies as Baseline 3, while utilizing the same task sensing and caching decisions as our proposed algorithm.

\item\textbf{Baseline 5}: Employs our proposed algorithm for subchannel and task allocations but always re-senses tasks without any caching strategy.
\end{itemize}
	\begin{table}
	\caption{\\
 S{\footnotesize YSTEM} P{\footnotesize ARAMETERS} S{\footnotesize ETTING}}
	\footnotesize
    \setlength{\tabcolsep}{3mm}{
    \begin{tabular}{|c|c|}
   \hline
   \textbf{System parameters} & \textbf{Values}   \\
   \hline  
   The number of users $K$ & 25 to 45  \\
   \hline 
   The number of subchannels $N$ & 6 to 26   \\
   \hline 
   Bandwidth of each subchannel $W$ & 1 MHz\\
   \hline
   Maximum AoI value $\Delta_{\rm max}$ & $50$\\
   \hline
   Cache capacity $Q$ & $10^7$ to $10^8$ bit\\
   \hline 
   Noise power density $N_0$ &	 $- 174$ dBm/Hz\\
   \hline 
   Distance between user and BS & 	30 to 500 m\\
   \hline
   Weight of system latency $\beta_1$ &  1 \\
   \hline
   Weight of AoI $\beta_2$ &  0.1 \\
   \hline 
   Frequency of re-sense $\beta_0$ & 0.7 \\
   \hline 
   Large-scale fading & $128.1 + 37.6 \log_{10} (x)$ \\
   \hline 
   Small-scale fading & Rayleigh fading, variance of 1  \\
   \hline 
   Sensing data rate $o_k^t$ & $10^4$ to $10^6$ bit/s\\
   \hline 
   Sensing energy consumption per bit & $10^{-12}$ to $10^{-11}$ J/bit \\
   \hline 
   Transmit power $P_k^t$ & 0.1 to 0.2 W  \\
   \hline 
   Maximum energy consumption $E_k$ & 0.01 to 0.1 J\\
   \hline 
   Smallest task bit size $V_{i_t}$ & $0.5 \times 10^7$ to $1.5 \times 10^{7}$ bit\\
   \hline
    \end{tabular}
}
     \end{table}
     
Fig. \ref{fig:2} shows the weighted sum of delay and AoI versus transmit power, where users number $K=30$ and subchannels $N=20$. Our proposed algorithm demonstrates superior performance over all baselines, achieving improvements of 96.56\%, 87.26\%, 75.53\%, 70.42\%, and 24.32\% , respectively, when transmit power is 0.1 Watt. Baseline 1 employs the uniform task allocation scheme, which may lead to a higher proportion of tasks being assigned to users with weaker sensing capabilities. Furthermore, Baseline 5 shows the best performance among baselines because it incorporates our subchannel and task allocation strategy.

Fig. \ref{fig:3} shows that the weighted sum of delay and AoI value decreases with more users due to enhanced user diversity, which allows for the selection of users with superior sensing and transmission capabilities. With the number of subchannels fixed at 20, our scheme consistently outperforms all baselines. Notably, the performance of Baseline 2 slightly deteriorates with increasing number of users. This is due to its random subchannel allocation, which may assign subchannels to users with poor capabilities. Baseline 1 demonstrates the most substantial improvement among all comparison baselines.
\begin{figure}[t]
\centering
\begin{minipage}[t]{0.47\linewidth}
\includegraphics[width=4.5cm]{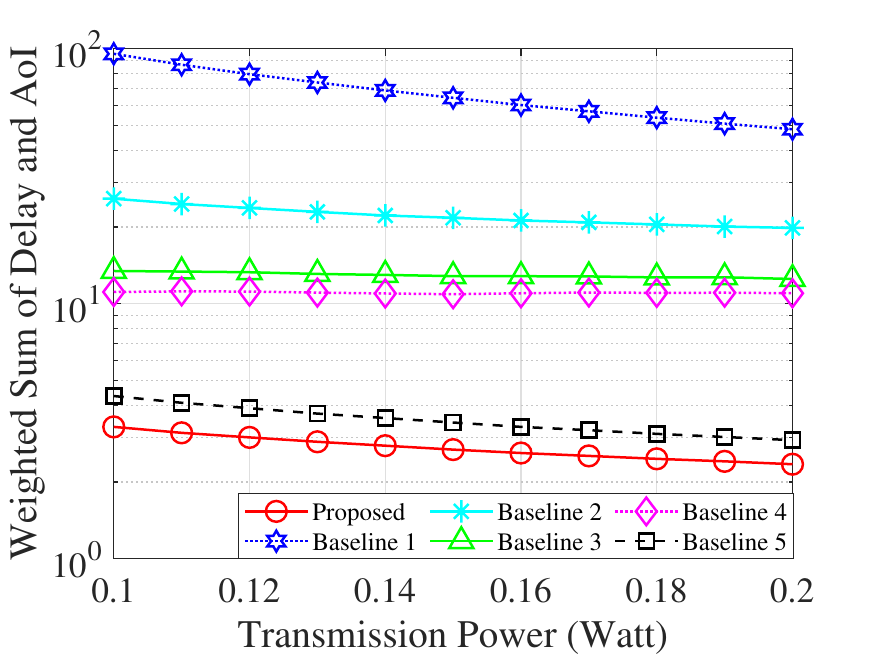}
\captionsetup{font=footnotesize}
\caption{Weighted sum of delay and AoI versus the transmit power.}
\label{fig:2}
\end{minipage}
\quad
\begin{minipage}[t]{0.47\linewidth}
\includegraphics[width=4.5cm]{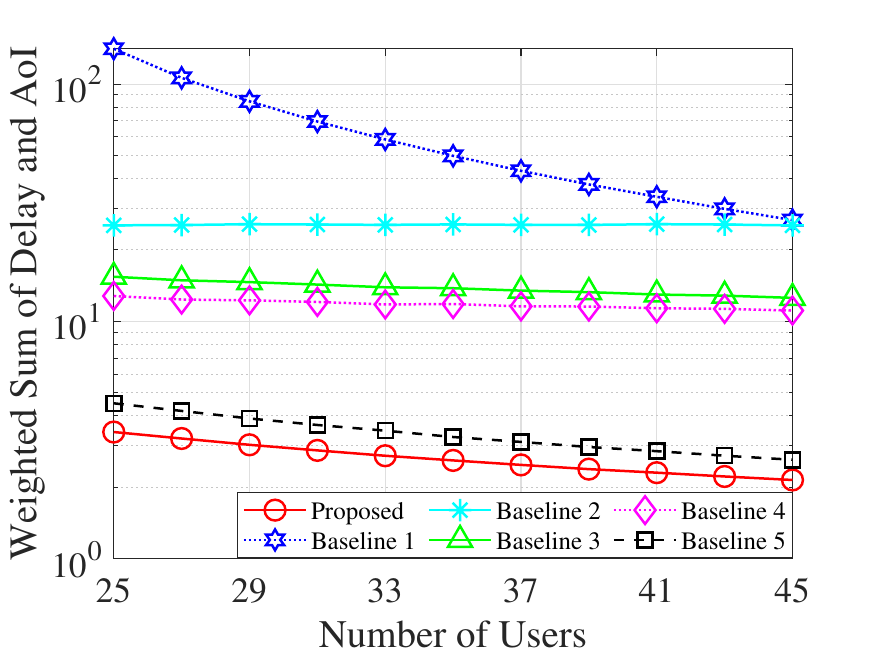}
\captionsetup{font=footnotesize}
\caption{Weighted sum of delay and AoI versus the number of users.}
\label{fig:3}
\end{minipage}
\centering
\end{figure}
\begin{figure}[t]
\centering
\begin{minipage}[t]{0.47\linewidth}
\includegraphics[width=4.5cm]{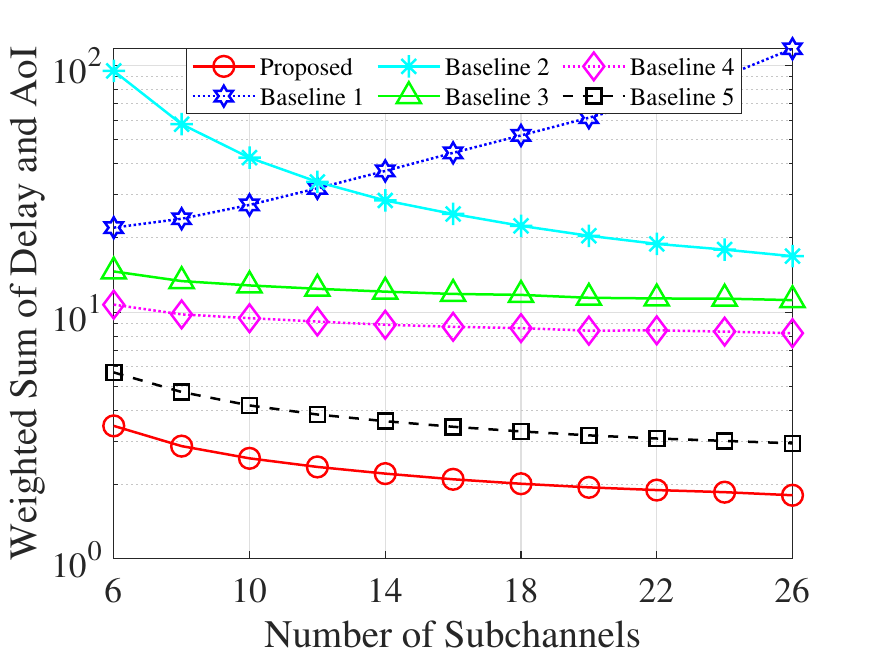}
\captionsetup{font=footnotesize}
\caption{Weighted sum of delay and AoI versus the number of subchannels.}
\label{fig:4}
\end{minipage}
\quad
\begin{minipage}[t]{0.47\linewidth}
\includegraphics[width=4.5cm]{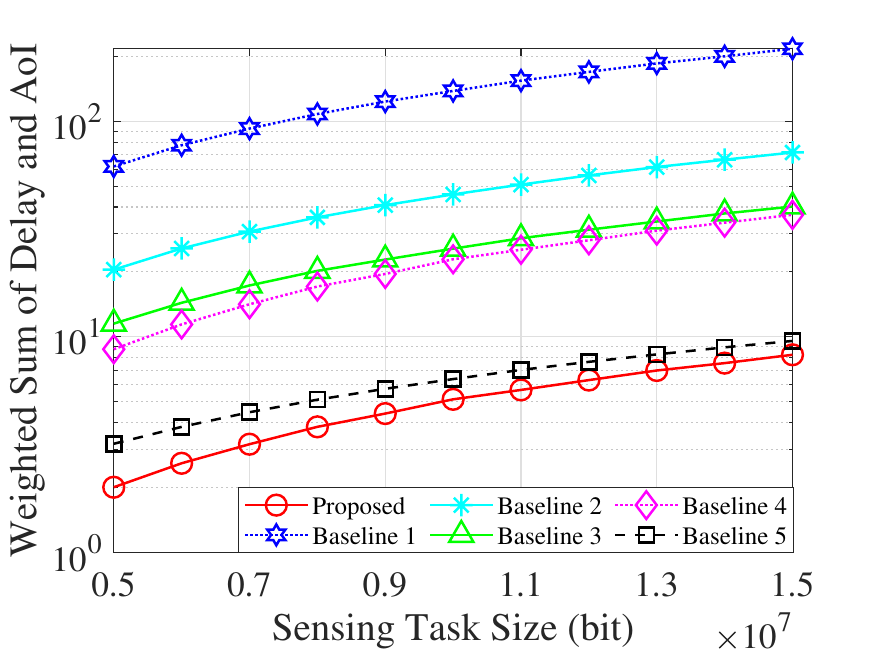}
\captionsetup{font=footnotesize}
\caption{Weighted sum of delay and AoI versus task bit size.}
\label{fig:5}
\end{minipage}
\centering 
\end{figure}

Fig. \ref{fig:4} evaluates the weighted sum of delay and AoI against subchannel numbers from 6 to 26, where the number of users is fixed at 30. Our proposed algorithm maintains superior performance throughout all subchannel ranges. Baseline 1 increases significantly with increasing subchannels due to its uniform task allocation, which may assign more users with lower sensing capacity to available subchannels. Moreover, Baseline 2 shows the most significant degradation, as its random channel allocation increases the likelihood of assigning high-capacity users to available subchannels.

Fig. \ref{fig:5} illustrates the weighted sum of delay and AoI performance versus sensing task bit size from $0.5\times10^7$ to $1.5\times10^7$ bits. All schemes demonstrate linear growth with task size, but our proposed method maintains the lowest values throughout. We can observe that Baseline 1 performs the poorest performance due to its uniform task allocation, indicating that task allocation strategy has a more significant impact than subchannel allocation on system performance.

\section{Conclusion}
In this study, we addressed the challenge of minimizing system latency while maintaining the freshness of cached tasks in cache-enabled MCS networks. By decomposing the non-convex mixed-integer programming problem into multiple one-shot sub-problems, we developed a framework that optimized user selection, subchannel allocation, task allocation, sensing decisions, and caching policies. Based on this framework, we proposed a time-efficient algorithm that sequentially solved these one-shot problems across the entire time horizon. Extensive simulations demonstrated the effectiveness of our algorithm. 

\bibliographystyle{IEEEtran}

\bibliography{Joint_System_Latency_and_Data_Freshness_Optimization_for_Cache-enabled_Mobile_Crowdsensing_Networks.bib}

\end{document}